\documentclass{article}
\usepackage{kr}

\usepackage{times}
\usepackage{soul}
\usepackage{url}
\usepackage[hidelinks]{hyperref}
\usepackage[utf8]{inputenc}
\usepackage[small]{caption}
\usepackage{graphicx}
\usepackage{amsmath}
\usepackage{amsthm}
\usepackage{booktabs}
\usepackage{mathtools}
\usepackage{amssymb}
\usepackage{pifont}
\usepackage[linesnumbered,ruled,vlined]{algorithm2e}
\usepackage{tikz}
\usepackage[capitalize,noabbrev]{cleveref}
\usepackage{multirow}
\usepackage{colortbl}
\usepackage{complexity}
\usepackage{siunitx}
\usepackage{forest}
\usepackage[inline]{enumitem}
\usepackage{subcaption}

\usetikzlibrary{cd}
\usetikzlibrary{shapes}
\usetikzlibrary{fit}

\SetKwProg{Fn}{Function}{:}{}

\SetKwData{functionCall}{call}
\SetKwData{functionSignature}{signature}
\SetKwData{functionNames}{F}
\SetKwData{domainNames}{D}
\SetKwData{newDomainNames}{D'}
\SetKwData{expression}{e}
\SetKwData{functions}{funs}
\SetKwData{args}{args}

\SetKwFunction{newFunctionName}{newFunctionName}
\SetKwFunction{newDomainName}{newDomainName}
\SetKwFunction{visit}{v}
\SetKwFunction{actuallyVisit}{visit}

\newtheorem{theorem}{Theorem}

\newtheorem{conjecture}{Conjecture}

\newtheorem{example}{Example}
\newtheorem{definition}{Definition}

\title{Synthesising Recursive Functions for First-Order Model Counting:\\
  Challenges, Progress, and Conjectures}

\author{%
Paulius Dilkas$^1$\and
Vaishak Belle$^2$\\
\affiliations
$^1$National University of Singapore, Singapore, Singapore\\
$^2$University of Edinburgh, Edinburgh, UK\\
\emails
paulius.dilkas@nus.edu.sg,
vbelle@ed.ac.uk
}

\crefname{algocf}{Algorithm}{Algorithms}
\Crefname{algocf}{Algorithm}{Algorithms}
\crefname{condition}{Condition}{Conditions}
\Crefname{condition}{Condition}{Conditions}
\crefname{line}{line}{lines}

\crefalias{diagram}{equation}
\crefname{diagram}{Diagram}{Diagrams}
\Crefname{diagram}{Diagram}{Diagrams}
\creflabelformat{diagram}{#2\textup{(#1)}#3}

\crefalias{formula}{equation}
\crefname{formula}{Formula}{Formulas}
\Crefname{formula}{Formula}{Formulas}
\creflabelformat{formula}{#2\textup{(#1)}#3}

\newcommand\pfun{\mathrel{\ooalign{\hfil$\mapstochar\mkern5mu$\hfil\cr$\to$\cr}}}
\newcommand{\cmark}{\ding{51}}%
\newcommand{\xmark}{\ding{55}}%
\newcommand{\FOtwo}{$\mathsf{FO}^{2}$}
\newcommand{\FOthree}{$\mathsf{FO}^{3}$}
\newcommand{\SFO}{$\mathsf{S}^{2}\mathsf{FO}^{2}$}
\newcommand{\SRU}{$\mathsf{S}^{2}\mathsf{RU}$}
\newcommand{\Uone}{$\mathsf{U}_{1}$}
\newcommand{\Ctwo}{$\mathsf{C}^{2}$}
\newcommand{\IFO}{$\mathsf{I}\mathsf{FO}^{2}$}
\newcommand\mdoubleplus{\mathbin{+\mkern-10mu+}}
\newcommand{\Done}{\domainNames_{1}}
\newcommand{\Dtwo}{\domainNames_{2}}
\newcommand{\Dthree}{\domainNames_{3}}
\newcommand{\predicate}{\texttt{\textup{p}}}
\newcommand{\predicateq}{\texttt{\textup{q}}}
\newcommand{\predicates}{\texttt{\textup{s}}}
\newcommand{\predicater}{\texttt{\textup{r}}}
\newcommand{\nulll}{\texttt{\textup{null}}}

\DeclareMathOperator{\friends}{\texttt{\textup{friends}}}
\DeclareMathOperator{\smokes}{\texttt{\textup{smokes}}}
\DeclareMathOperator{\CR}{\textsc{CR}}
\DeclareMathOperator{\GDR}{\textsc{GDR}}
\DeclareMathOperator{\Reff}{\textsc{Ref}}
\DeclareMathOperator{\dom}{dom}
\DeclareMathOperator{\Doms}{Doms}
\DeclareMathOperator{\Vars}{Vars}
\DeclareMathOperator{\Preds}{Preds}
\SetKwProg{Fn}{Function}{:}{}
\SetKwFunction{identifyRecursion}{r}
\SetKwFunction{generateMaps}{genMaps}

\makeatletter
\newcommand{\nosemic}{\renewcommand{\@endalgocfline}{\relax}}% Drop semi-colon ;
\newcommand{\dosemic}{\renewcommand{\@endalgocfline}{\algocf@endline}}% Reinstate semi-colon ;
\makeatother

\forestset{
sn edges/.style={for tree={edge={-Latex}}}
}

\begin{document}

\maketitle

\begin{abstract}
  First-order model counting (FOMC) is a computational problem that asks to
  count the models of a sentence in finite-domain first-order logic. In this
  paper, we argue that the capabilities of FOMC algorithms to date are limited
  by their inability to express many types of recursive computations. To enable
  such computations, we relax the restrictions that typically accompany domain
  recursion and generalise the circuits used to express a solution to an FOMC
  problem to directed graphs that may contain cycles. To this end, we adapt the
  most well-established (weighted) FOMC algorithm \textsc{ForcLift} to work with
  such graphs and introduce new compilation rules that can create cycle-inducing
  edges that encode recursive function calls. These improvements allow the
  algorithm to find efficient solutions to counting problems that were
  previously beyond its reach, including those that cannot be solved efficiently
  by any other exact FOMC algorithm. We end with a few conjectures on what
  classes of instances could be domain-liftable as a result.
\end{abstract}

\section{Introduction}\label{sec:introduction}

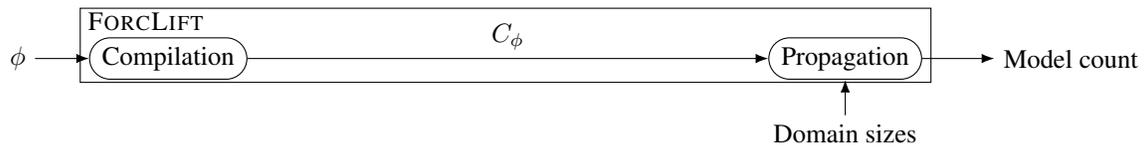
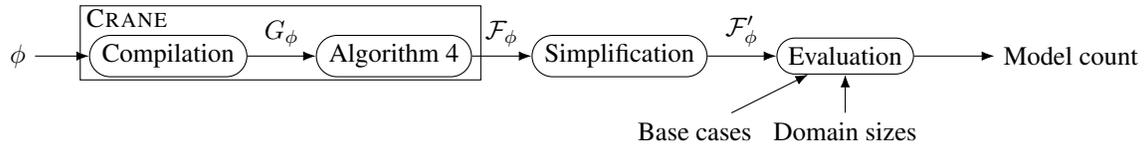
\begin{figure*}[t]
  \centering
  \begin{subfigure}{\textwidth}
    \centering
    \begin{tikzpicture}
      \node at (-2, 0) (start) {$\phi$};
      \node[draw,rounded rectangle] at (0, 0) (kc) {Compilation};
      \node[draw,rounded rectangle] at (9, 0) (evaluation) {Propagation};
      \node at (12, 0) (end) {Model count};
      \node at (9, -1) (sizes) {Domain sizes};
      \node at (-0.3, 0.5) {\textsc{ForcLift}};
      \node[draw,fit={(kc) (evaluation)},inner ysep=6pt,yshift=5pt] {};
      \draw[-Latex] (start) -- (kc);
      \draw[-Latex] (kc) -- node[above] {$C_{\phi}$} (evaluation);
      \draw[-Latex] (evaluation) -- (end);
      \draw[-Latex] (sizes) -- (evaluation);
    \end{tikzpicture}
    \caption{\textsc{ForcLift} compiles $\phi$ into a circuit $C_{\phi}$. The
      domain sizes are then used to propagate values through $C_{\phi}$,
      computing the model count.}
  \end{subfigure}
  \begin{subfigure}{\textwidth}
    \centering
    \begin{tikzpicture}
      \node at (-2, 0) (start) {$\phi$};
      \node[draw,rounded rectangle] at (0, 0) (kc) {Compilation};
      \node[draw,rounded rectangle] at (3, 0) (interpretation) {\Cref{alg:interpretation}};
      \node[draw,rounded rectangle] at (6, 0) (simplification) {Simplification};
      \node[draw,rounded rectangle] at (9, 0) (evaluation) {Evaluation};
      \node at (12, 0) (end) {Model count};
      \node at (7, -1) (base) {Base cases};
      \node at (9, -1) (sizes) {Domain sizes};
      \node at (-0.55, 0.5) {\textsc{Crane}};
      \node[draw,fit={(kc) (interpretation)},inner ysep=6pt,yshift=5pt] {};
      \draw[-Latex] (start) -- (kc);
      \draw[-Latex] (kc) -- node[above] {$G_{\phi}$} (interpretation);
      \draw[-Latex] (interpretation) -- node[above] {$\mathcal{F}_{\phi}$} (simplification);
      \draw[-Latex] (simplification) -- node[above] {$\mathcal{F}_{\phi}'$} (evaluation);
      \draw[-Latex] (evaluation) -- (end);
      \draw[-Latex] (base) -- (evaluation);
      \draw[-Latex] (sizes) -- (evaluation);
    \end{tikzpicture}
    \caption{\textsc{Crane} compiles $\phi$ into a graph $G_{\phi}$ and then
      converts $G_{\phi}$ into a collection of functions $\mathcal{F}_{\phi}$.
      In some cases, these functions can benefit from algebraic simplification.
      If some functions are defined recursively, base cases need to be
      established. One can then compute the model count of $\phi$ by running the
      main function in $\mathcal{F}'_{\phi}$ with domain sizes as
      arguments.}\label{fig:processcrane}
  \end{subfigure}
  \caption{A comparison of how \textsc{ForcLift} and \textsc{Crane} can be used
    to compute the model count of a formula $\phi$}\label{fig:process}
\end{figure*}

% 1. What is the problem?
% 2. Why is it interesting and important?

\emph{First-order model counting} (FOMC) is the problem of computing the number
of models of a sentence in first-order logic given the size(s) of its
domain(s)~\cite{DBLP:conf/pods/BeameBGS15}. \emph{Symmetric weighted FOMC}
(WFOMC) extends FOMC with (pairs of) weights on predicates and asks for a
weighted sum across all models instead. By fixing the sizes of the domains, a
WFOMC instance can be rewritten as an instance of (propositional) weighted model
counting~\cite{DBLP:journals/ai/ChaviraD08}. WFOMC emerged as the dominant
approach to \emph{lifted (probabilistic) inference}. Lifted inference techniques
exploit symmetries in probabilistic models by reasoning about sets rather than
individuals~\cite{DBLP:conf/ecai/Kersting12}. By doing so, many instances become
solvable in polynomial time~\cite{DBLP:conf/nips/Broeck11}. The development of
lifted inference algorithms coincided with work on probabilistic relational
models that combine the syntactic power of first-order logic with probabilistic
models such as Bayesian and Markov networks, allowing for a more relational view
of uncertainty
modelling~\cite{DBLP:series/synthesis/2016Raedt,DBLP:journals/ml/KimmigMG15,DBLP:journals/ml/RichardsonD06}.
Lifted inference techniques for probabilistic databases have also been inspired
by
WFOMC~\cite{DBLP:journals/pvldb/GatterbauerS15,DBLP:journals/debu/GribkoffSB14}.
While WFOMC has received more attention in the literature, FOMC is an
interesting problem in and of itself because of its connections to finite model
theory~\cite{DBLP:conf/kr/BremenK21} and applications in enumerative
combinatorics~\cite{DBLP:conf/ilp/BarvinekB0ZK21}.

% Totis et
% al.~\shortcite{DBLP:journals/jair/TotisDRK23} proposed an FOMC-based approach to
% counting solutions to constraint satisfaction problems.

% 3. Why is it hard? (E.g., why do naive approaches fail?)
% complexity/liftability -> murky boundary

Traditionally in computational complexity theory, a problem is \emph{tractable}
if it can be solved in time polynomial in the instance size. The equivalent
notion in (W)FOMC is liftability. A (W)FOMC instance is \emph{(domain-)liftable}
if it can be solved in time polynomial in the size(s) of the
domain(s)~\cite{DBLP:conf/starai/JaegerB12}. Many classes of instances are known
to be liftable. First, Van den Broeck~\shortcite{DBLP:conf/nips/Broeck11} showed
that the class of all sentences of first-order logic with up to two variables
(denoted \FOtwo{}) is liftable. Then Beame et
al.~\shortcite{DBLP:conf/pods/BeameBGS15} proved that there exists a sentence
with three variables for which FOMC is $\#\P_{1}$-complete (i.e., \FOthree{} is
not liftable). Since these two seminal results, most research on (W)FOMC focused
on developing faster solutions for the \FOtwo{}
fragment~\cite{DBLP:conf/uai/BremenK21,DBLP:conf/aaai/MalhotraS22} and defining
new liftable fragments. These fragments include \SFO{} and
\SRU{}~\cite{DBLP:conf/nips/KazemiKBP16},
\Uone{}~\cite{DBLP:conf/lics/KuusistoL18}, \Ctwo{} (i.e., the two-variable
fragment with counting
quantifiers)~\cite{DBLP:journals/jair/Kuzelka21,DBLP:conf/aaai/MalhotraS22}, and
\Ctwo{} extended with axioms for trees~\cite{DBLP:conf/kr/BremenK21}. On the
empirical front, there are several implementations of exact WFOMC algorithms:
\textsc{Alchemy}~\cite{DBLP:journals/cacm/GogateD16},
\textsc{FastWFOMC}~\cite{DBLP:conf/uai/BremenK21},
\textsc{ForcLift}~\cite{DBLP:conf/ijcai/BroeckTMDR11}, and
\textsc{L2C}~\cite{DBLP:conf/kr/KazemiP16}. Approximate counting is supported by
\textsc{Alchemy}, \textsc{ApproxWFOMC}~\cite{DBLP:conf/ijcai/BremenK20},
\textsc{ForcLift}~\cite{DBLP:conf/uai/BroeckCD12},
\textsc{Magician}~\cite{DBLP:conf/aaai/VenugopalSG15}, and
\textsc{Tuffy}~\cite{DBLP:journals/pvldb/NiuRDS11}.

% 4. Why hasn't it been solved before? (Or, what's wrong with previous proposed
% solutions? How does mine differ?)
% simple unliftable instances -> recursion

We claim that the capabilities of (W)FOMC algorithms can be significantly
extended by empowering them with the ability to construct recursive solutions.
The topic of recursion in the context of WFOMC has been studied before but in
limited ways. Barv{\'{\i}}nek et al.~\shortcite{DBLP:conf/ilp/BarvinekB0ZK21}
use WFOMC to generate numerical data that is then used to conjecture recurrence
relations that explain that data. Van den
Broeck~\shortcite{DBLP:conf/nips/Broeck11} introduced the idea of \emph{domain
  recursion}. Intuitively, domain recursion partitions a domain of size $n$ into
a single explicitly named constant and the remaining domain of size $n-1$.
However, many stringent conditions are enforced to ensure that the search for a
tractable solution always terminates.

% 5. What are the key components of my approach and results? Also include any
% specific limitations.

% 5.1) An overview of how ForcLift/Crane work.

In this work, we show how to relax these restrictions in a way that results in
an algorithm capable of handling more instances in a lifted manner. The ideas
presented in this paper are implemented in
\textsc{Crane}\footnote{\url{https://github.com/dilkas/crane}}\textsuperscript{,}\footnote{\url{https://doi.org/10.5281/zenodo.8004077}}---an
extension of the (W)FOMC algorithm \textsc{ForcLift}. The differences in how
these algorithms operate are depicted in \cref{fig:process}. Compilation is
performed by applying various \emph{(compilation) rules} to the input (or some
derivative) formula, gradually constructing a circuit (in the case of
\textsc{ForcLift}) or a graph (in the case of \textsc{Crane}). \textsc{ForcLift}
applies compilation rules via greedy search, whereas \textsc{Crane} also
supports a hybrid search algorithm that applies some rules greedily and some
using breadth-first search.\footnote{In the current implementation, the rules
  applied non-greedily are: atom counting, inclusion-exclusion, independent
  partial groundings, Shannon decomposition, shattering, and two new rules
  described in \cref{sec:dr,sec:ref}. See previous work for more information
  about the rules~\cite{DBLP:conf/ijcai/BroeckTMDR11} and the search
  algorithm~\cite{dilkas2023generalising}.} This alternative was introduced
because there is no reason to expect greedy search to be complete. Another
difference is that---in \textsc{Crane}---the product of compilation is not
directly evaluated but transformed into a collection of functions on domain
sizes. Hence, our approach is reminiscent of previous work on lifted inference
via compilation to C++ programs~\cite{DBLP:conf/kr/KazemiP16} and the broader
area of functional
synthesis~\cite{DBLP:conf/cav/GoliaRM20,DBLP:conf/pldi/KuncakMPS10}.

% 5.2) The main idea: cycles that represent recursive calls.

Using labelled directed graphs instead of circuits enables \textsc{Crane} to
construct recursive solutions by representing recursive function calls via
cycle-inducing edges. A hypothetical instance of compilation could proceed as
follows. Suppose the input formula $\phi$ depends on a domain of size
$n \in \mathbb{N}_{0}$. \emph{Generalised domain recursion} (GDR)---one of the
new compilation rules---transforms $\phi$ into a different formula $\psi$ with
an additional constant and some \emph{constraints}. After some more
transformations, the constraints in $\psi$ can be removed, replacing the domain
of size $n$ with a new domain of size $n-1$---this is the responsibility of the
\emph{constraint removal} (CR) compilation rule. Afterwards, another compilation
rule recognises that the resulting formula matches the input formula $\phi$
except for referring to a different domain. This observation allows us to add a
cycle-forming edge to the graph, which can be interpreted as a function $f$
relying on $f(n-1)$ to compute $f(n)$.

% 5.3) contributions, references to sections, shortcomings

We begin by introducing some notation, terminology, and the problem of FOMC in
\cref{sec:recprelims}. Then, in \cref{sec:methods}, we define the graphs that
replace circuits in representing a solution to such a problem. \Cref{sec:rules}
introduces the new compilation rules. \Cref{sec:interpret} describes an
algorithm that converts such a graph into a collection of (potentially
recursive) functions. \Cref{sec:results} compares \textsc{ForcLift} and
\textsc{Crane} on various counting problems. We show that:
\begin{enumerate*}[label=(\roman*)]
  \item \textsc{Crane} performs as well as \textsc{ForcLift} on the instances
  that were already solvable by \textsc{ForcLift},
  \item \textsc{Crane} is also able to handle most of the instances that
  \textsc{ForcLift} fails on, \emph{including those outside of currently-known
    domain-liftable fragments} such as \Ctwo{}.
\end{enumerate*}
Finally, \cref{sec:conclusion} outlines some conjectures and directions for
future work.

\section{Preliminaries}\label{sec:recprelims}

Our representation of FOMC instances builds on the format used by
\textsc{ForcLift}, some aspects of which are described by Van den Broeck et
al.~\shortcite{DBLP:conf/ijcai/BroeckTMDR11}. \textsc{ForcLift} can translate
sentences in a variant of function-free many-sorted first-order logic with
equality to this internal format. We use lowercase Latin letters for predicates
(e.g., $\predicate$) and constants (e.g., $x$), uppercase Latin letters for
variables (e.g., $X$), and uppercase Greek letters for domains (e.g., $\Delta$).
Sometimes we write predicate $\predicate$ as $\predicate/n$, where
$n \in \mathbb{N}^{+}$ is the \emph{arity} of $\predicate$. An \emph{atom} is
$\predicate(x_1, \dots, x_n)$ for some predicate $\predicate/n$ and terms
$x_{1}, \dots, x_{n}$. A \emph{term} is either a constant or a variable. A
\emph{literal} is either an atom or the negation thereof (denoted by
$\neg \predicate(x_1, \dots, x_n)$). Let $\mathcal{D}$ be the set of all
relevant domains. Initially, $\mathcal{D}$ contains all domains mentioned by the
input formula. During compilation, new domains are added to $\mathcal{D}$. Each
new domain is interpreted as a subset of another domain in $\mathcal{D}$.

We write $\langle\cdot\rangle$ for lists, $|\cdot|$ for the length of a list,
and $\mdoubleplus$ for list concatenation. We write $\pfun$ to denote partial
functions and $\dom(\cdot)$ for the domain of a function. Let $\Doms$
(respectively, $\Vars$) be the function that maps any expression to the set of
domains (respectively, variables) used in it. Let $S$ be a set of constraints or
literals, $V$ a set of variables, and $x$ a term. We write $S[x/V]$ to denote
$S$ with all occurrences of all variables in $V$ replaced with $x$.

\begin{definition}[Constraint]\label{def:constraint}
  An \emph{(inequality) constraint} is a pair $(a, b)$, where $a$ is a variable,
  and $b$ is a term. It constrains $a$ and $b$ to be different.
\end{definition}

\begin{definition}[Clause]\label{def:clause}
  A \emph{clause} is $c = (L, E, \delta_c)$, where $L$ is a set of literals, $E$
  is a set of constraints, and $\delta_c$ is the domain map. \emph{Domain map}
  $\delta_{c}\colon \Vars(c) \to \mathcal{D}$ is a function that maps all
  variables in $c$ to their domains such that (s.t.) if $(X, Y) \in E$ for some
  variables $X$ and $Y$, then $\delta_c(X) = \delta_c(Y)$. For convenience, we
  sometimes write $\delta_c$ for the domain map of $c$ without unpacking $c$
  into its three constituents.
\end{definition}

\begin{definition}[Formula]\label{def:formula}
  A \emph{formula} is a set of clauses s.t.\ all constraints and atoms `type
  check' with respect to domains.
\end{definition}

\begin{example}\label{example:first}
  Let $\phi \coloneqq \{\, c_1, c_2 \,\}$ be a formula with clauses
  \begin{align*}
    c_1 &\coloneqq
          \begin{multlined}[t]
            (\{\, \neg \predicate(X, Y), \neg \predicate(X, Z) \,\}, \{\, (Y, Z) \,\}, \\
            \{\, X \mapsto \Gamma, Y \mapsto \Delta, Z \mapsto \Delta \,\}),
          \end{multlined}\\
    c_2 &\coloneqq
          \begin{multlined}[t]
            (\{\, \neg \predicate(X, Y), \neg \predicate(Z, Y) \,\}, \{\, (X, Z) \,\}, \\
            \{\, X \mapsto \Gamma, Y \mapsto \Delta, Z \mapsto \Gamma \,\})
          \end{multlined}
  \end{align*}
  for some predicate $\predicate/2$, variables $X$, $Y$, $Z$, and domains
  $\Gamma$ and $\Delta$. All variables that occur as the first argument to
  $\predicate$ are in $\Gamma$, and, likewise, all variables that occur as the
  second argument to $\predicate$ are in $\Delta$. Therefore, $\phi$ is a valid
  formula.
\end{example}

One can read such a formula as a sentence in first-order logic. All variables in
a clause are implicitly universally quantified, and all clauses in a formula are
implicitly linked by a conjunction. Thus, formula $\phi$ from
\cref{example:first} reads as
\begin{equation}\label[formula]{eq:formula}
  \begin{split}
    &\begin{multlined}[t]
      (\forall X \in \Gamma\text{. }\forall Y, Z \in \Delta\text{. }\\
      Y \ne Z \Rightarrow \neg \predicate(X, Y) \lor \neg \predicate(X, Z)) \land{}
    \end{multlined}\\
    &\begin{multlined}[t]
      (\forall X, Z \in \Gamma\text{. }\forall Y \in \Delta\text{. }\\
      X \ne Z \Rightarrow \neg \predicate(X, Y) \lor \neg \predicate(Z, Y)).
    \end{multlined}
  \end{split}
\end{equation}

There are two differences between \cref{def:constraint,def:clause,def:formula}
and the corresponding concepts by Van den Broeck et
al.~\shortcite{DBLP:conf/ijcai/BroeckTMDR11}.\footnote{Van den Broeck et
  al.~\shortcite{DBLP:conf/ijcai/BroeckTMDR11} refer to clauses and formulas as
  c-clauses and c-theories, respectively.} First, we decouple variable-to-domain
assignments from constraints and move them to a separate function $\delta_{c}$
in \cref{def:clause}. Second, while they allow for equality constraints and
constraints of the form $X \not\in \Delta$ for some variable $X$ and domain
$\Delta$, we exclude these constraints simply because they are inessential. Note
that if we replace $Y \ne Z$ in \cref{eq:formula} with $Y = Z$, then
\cref{eq:formula} can be simplified to have one fewer variables. Similarly, if
the same inequality is replaced by $Y = x$ for some constant $x$, then $Y$ can
be eliminated as well. Since constraints are always interpreted as preconditions
for the disjunction of literals in the clause (as in \cref{eq:formula}),
equality constraints can be eliminated without any loss in expressivity.

\begin{example}\label{example:simple}
  Let $\Delta$ be a domain of size $n \in \mathbb{N}_{0}$. The model count of
  $\forall X \in \Delta\text{. } \predicate(X) \lor \predicateq(X)$ is then
  $3^{n}$. Intuitively, since both predicates are of arity one, they can be
  interpreted as subsets of $\Delta$. Thus, the formula says that each element
  of $\Delta$ has to be in $\predicate$ or $\predicateq$ or both.
\end{example}

\begin{example}\label{example:smokers}
  Consider a variant of the well-known `friends and smokers' example
  $\forall X, Y \in \Delta\text{.
  } \smokes(X) \land \friends(X, Y) \Rightarrow \smokes(Y)$. Letting
  $n \coloneqq |\Delta|$ as before, the model count is
  $\sum_{k=0}^{n} \binom{n}{k}2^{n^{2} - k(n-k)}$~\cite{DBLP:conf/kr/BroeckMD14}.
\end{example}

Let $\sigma\colon \mathcal{D} \to \mathbb{N}_{0}$ be the \emph{domain size
  function} that maps each domain to a non-negative integer.

% By a relation between two (potentially distinct) sets, we mean any
% subset of their Cartesian product. Thus, if both sets have cardinality two,
% their Cartesian product has cardinality four, and this product set has 2^4 =
% 16 subsets. To check whether a subset S is a model, we interpret S as the set
% of all pairs (X, Y) for which p(X, Y) is true. Each of the two clauses in the
% example (c_1 and c_2) imposes certain constraints on S. For example, for
% domains Gamma = Delta = {1, 2}, clause c_1 requires that:
% 1) either (1, 1) or (1, 2) is not in S,
% 2) either (2, 1) or (2, 2) is not in S.

\begin{example}
  Let $\phi$ be as in \cref{example:first} and $\Gamma = \Delta = \{\,1, 2\,\}$,
  i.e., $\sigma(\Gamma) = \sigma(\Delta) = 2$. There are $2^{2 \times 2} = 16$
  possible relations between $\Gamma$ and $\Delta$, i.e., subsets of
  $\Gamma \times \Delta$. Let us count how many of them satisfy the conditions
  imposed on predicate $\predicate$. The empty relation does. All four relations
  of cardinality one (e.g., $\{\, (1, 1) \,\}$) do too. Finally, there are two
  relations of cardinality two---$\{\, (1, 1), (2, 2) \,\}$ and
  $\{\, (1, 2), (2, 1) \,\}$---that satisfy the conditions as well. Therefore,
  the FOMC of $(\phi, \sigma)$ is 7. Incidentally, the FOMC of $\phi$ counts
  partial injections from $\Gamma$ to $\Delta$. We will continue to use the
  problem of counting partial injections as the main running example.
\end{example}

\section{First-Order Computational Graphs}\label{sec:methods}

\begin{figure}[t]
  \centering
  \begin{forest}
    for tree={s sep=10mm, sn edges}
    [$\bigwedge_{x \in \Delta}$,ellipse,draw,label=right:{\color{blue}{$3^{n}$}}
    [$\lor$,ellipse,draw,label=right:{\color{blue}{$2 + 1 = 3$}}
    [$\land$,ellipse,draw,label=left:{\color{blue}{$2 \times 1 = 2$}}
    [$p(x) \lor \neg p(x)$,rectangle,draw,fill=gray!25,label=below:{\color{blue}{2}}]
    [$\land$,ellipse,draw,label={[label distance=0cm,text=blue]87:$1 \times 1 = 1$}
    [$q(x)$,rectangle,draw,fill=gray!25,label=below:{\color{blue}{1}}]
    [$\top$,rectangle,draw,fill=gray!25,label=below:{\color{blue}{1}}]
    ]
    ]
    [$\land$,ellipse,draw,label=right:{\color{blue}{$1 \times 1 = 1$}}
    [$\neg q(x)$,rectangle,draw,fill=gray!25,label=below:{\color{blue}{1}}]
    [$p(x)$,rectangle,draw,fill=gray!25,label=below:{\color{blue}{1}}]
    ]
    ]
    ]
  \end{forest}
  \caption{A circuit produced by \textsc{ForcLift} for \cref{example:simple}.
    Values and computations in blue (on the outside of each node) show how a
    bottom-up evaluation of the circuit computes the model
    count.}\label{fig:simplecircuit}
\end{figure}
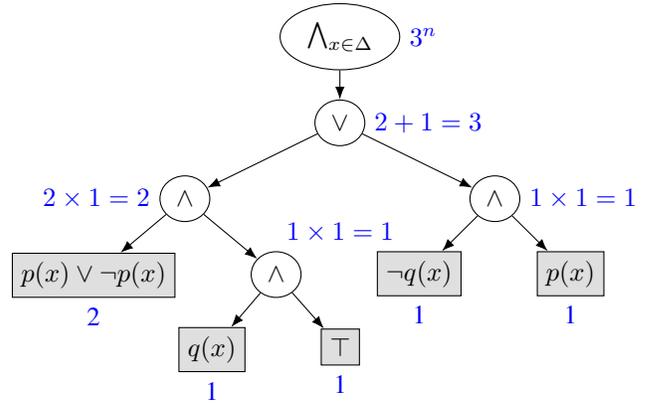

Darwiche~\shortcite{DBLP:journals/jancl/Darwiche01} introduced
\emph{deterministic decomposable negation normal form} (d-DNNF) circuits for
propositional knowledge compilation and showed that the model count of a
propositional formula can be computed in time linear in the size of the circuit.
Van den Broeck et al.~\shortcite{DBLP:conf/ijcai/BroeckTMDR11} generalised them
to first-order logic via \emph{first-order d-DNNF (FO d-DNNF) circuits}. FO
d-DNNF circuits (hereafter called \emph{circuits}) are directed acyclic graphs
with nodes corresponding to formulas in first-order logic---see
\cref{fig:simplecircuit} for an example. The following types of nodes are
supported by \textsc{ForcLift}: caching ($\Reff$), contradiction ($\bot$),
tautology ($\top$), decomposable conjunction ($\land$), decomposable
set-conjunction ($\bigwedge$), deterministic disjunction ($\lor$), deterministic
set-disjunction ($\bigvee$), domain recursion, grounding, inclusion-exclusion,
smoothing, and unit clause. We refer the reader to previous
work~\cite{DBLP:conf/nips/Broeck11,DBLP:conf/ijcai/BroeckTMDR11} for more
information about node types and their interpretations for computing the
(W)FOMC\@.

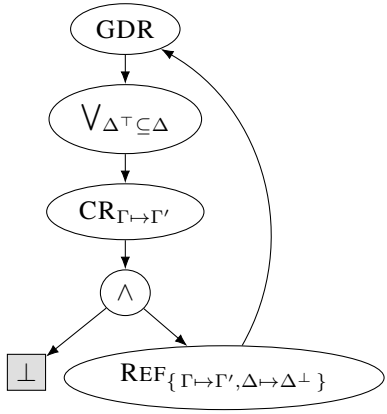
\begin{figure}[t]
  \centering
  \begin{forest}
    for tree={sn edges}
    [$\GDR$,ellipse,draw,name=gdr
    [$\bigvee_{\Delta^\top \subseteq \Delta}$,ellipse,draw
    [$\CR_{\Gamma \mapsto \Gamma'}$,ellipse,draw
    [$\land$,ellipse,draw
    [$\bot$,rectangle,draw,fill=gray!25]
    [$\Reff_{\{\, \Gamma \mapsto \Gamma', \Delta \mapsto \Delta^{\bot} \,\}}$,ellipse,draw,name=ref]
    ]
    ]
    ]
    ]
    \draw[-Latex,bend right=45] (ref) to (gdr);
  \end{forest}
  \caption{A simplified version of the FCG constructed by \textsc{Crane} for the
    problem of counting partial injections from \cref{example:first}. Here we
    omit some parameters as well as nodes whose only arithmetic effect is
    multiplication by one.}\label{fig:examplefcg}
\end{figure}

We introduce \emph{first-order computational graphs} (FCGs) that generalise
circuits by dispensing with acyclicity. An FCG is a (weakly connected) directed
graph with a single source, node labels, and ordered outgoing edges. Node labels
consist of two parts: the \emph{type} and the \emph{parameters}. The type of a
node determines its out-degree. We make the following changes to the node types
already supported by \textsc{ForcLift}. First, we introduce a new type for
constraint removal ($\CR$). Second, we replace domain recursion with generalised
domain recursion ($\GDR$). And third, for reasoning about partially-constructed
FCGs, we write $\star$ for a placeholder type that is yet to be replaced. We
write $T_p$ for an FCG that has a node with label $T_p$ (i.e., type $T$ and
parameter(s) $p$) and $\star$'s as all of its direct successors. We also write
$T_p(v)$ for an FCG with one edge from a node labelled $T_{p}$ to another node
$v$. See \cref{fig:examplefcg} for an example FCG\@.

% Connecting FCGs with formulas
Finally, we introduce a structure that represents a solution while it is still
being built. A \emph{chip} is a pair $(G, L)$, where $G$ is an FCG, and $L$ is a
list of formulas, s.t. $|L|$ is equal to the number of $\star$'s in $G$. $L$
contains formulas that still need to be compiled. Once a formula is compiled, it
replaces one of the $\star$'s in $G$ according to a set order. We call an FCG
\emph{complete} (i.e., it represents a \emph{complete solution}) if it has no
$\star$'s. Similarly, a chip is complete if its FCG is complete.

\section{New Compilation Rules}\label{sec:rules}

A \emph{(compilation) rule} takes a formula and returns a set of chips. The
cardinality of this set is the number of ways the rule can be applied to the
input formula. While \textsc{ForcLift}~\cite{DBLP:conf/ijcai/BroeckTMDR11}
heuristically chooses one of them, in an attempt to not miss a solution,
\textsc{Crane} returns them all. In particular, if a rule returns an empty set,
then that rule does not apply to the formula.

\subsection{Generalised Domain Recursion}\label{sec:dr}

% Main idea
The main idea behind domain recursion (the original version by Van den
Broeck~\shortcite{DBLP:conf/nips/Broeck11} and the one presented here) is as
follows. Let $\Omega \in \mathcal{D}$ be a domain. Assuming that
$\Omega \ne \emptyset$, pick some $x \in \Omega$. Then, for every variable
$X \in \Omega$ that occurs in a literal, consider two possibilities: $X = x$ and
$X \ne x$.

\begin{example}\label{example:dr}
  Let $\phi$ be a formula with a single clause
  \begin{multline*}
    (\{\, \neg \predicate(X, Y), \neg \predicate(X, Z) \,\}, \{\, (Y, Z) \,\}, \\
    \{\, X \mapsto \Gamma, Y \mapsto \Delta, Z \mapsto \Delta \,\}).
  \end{multline*}
  Then we can introduce constant $x \in \Gamma$ and rewrite $\phi$ as
  $\phi' = \{\, c_{1}, c_{2} \,\}$, where
  \begin{align*}
    c_{1} &\coloneqq \begin{multlined}[t]
      (\{\, \neg \predicate(x, Y), \neg \predicate(x, Z) \,\}, \{\, (Y, Z) \,\}, \\
      \{\, Y \mapsto \Delta, Z \mapsto \Delta \,\}),
      \end{multlined}\\
    c_{2} &\coloneqq \begin{multlined}[t]
      (\{\, \neg \predicate(X, Y), \neg \predicate(X, Z) \,\}, \{\, (X, x), (Y, Z) \,\}, \\
      \{\, X \mapsto \Gamma', Y \mapsto \Delta, Z \mapsto \Delta \,\}),
      \end{multlined}
  \end{align*}
  and $\Gamma' \coloneqq \Gamma \setminus \{\, x \,\}$.
\end{example}

% Original domain recursion
Van den Broeck~\shortcite{DBLP:conf/nips/Broeck11} imposes stringent conditions
on the input formula to ensure that its expanded version (as in
\cref{example:dr}) can be handled efficiently. \Cref{example:first} cannot be
handled by \textsc{ForcLift} because there is no root binding class, i.e., the
two root variables belong to different equivalence classes with respect to the
binding relationship.

% The clauses in this expanded
% formula are then partitioned into three parts based on whether the
% transformation introduced constants, constraints, or both. The conditions ensure
% that these parts can be treated independently.

\begin{algorithm}[t]
  \caption{The compilation rule for $\GDR$ nodes}\label{alg:domainrecursion}
  \KwIn{formula $\phi$, set of all relevant domains $\mathcal{D}$}
  \KwOut{set of chips $S$}
  $S \gets \emptyset$\;
  \ForEach{domain $\Omega \in \mathcal{D}$ s.t.\ there is $c \in \phi$ and $X \in \Vars(L_c)$ s.t. $\delta_c(X) = \Omega$\label{line:condition}}{
    $\phi' \gets \emptyset$\;
    $x \gets \text{a new constant in domain } \Omega$\;\label{line:constant}
    \ForEach{clause $c = (L, E, \delta) \in \phi$\label{line:forclause}}{
      $V \gets \{\, X \in \Vars(L) \mid \delta(X) = \Omega \,\}$\;\label{line:V}
      \ForEach{$W \subseteq V$ s.t. $W^2 \cap E = \emptyset$ {\bf and} $W \cap \{\, X \in \Vars(E) \mid (X, y) \in E \text{ for some constant } y \,\} = \emptyset$\label{line:conditions}}{
        \tcc{$\delta'$ restricts $\delta$ to the new set of variables}
        $\phi' \gets \phi' \cup \{\, (L[x/W], E[x/W] \cup \{\, (X, x) \mid (X \in V \setminus W) \,\}, \delta') \,\}$\;\label{line:generation}
      }
    }
    $S \gets S \cup \{\, (\GDR, \langle\phi'\rangle) \,\}$\;
  }
\end{algorithm}

% Description of GDR, in contrast to DR
In contrast, GDR has only one precondition: for GDR to be applicable to domain
$\Omega \in \mathcal{D}$, there must be at least one variable with domain
$\Omega$ featured in a literal (and not just in constraints). Without such
variables, GDR would have no effect on the formula. The expanded formula is then
left as-is to be handled by other compilation rules. Typically, after a few more
rules are applied, a combination of $\CR$ and $\Reff$ nodes introduces a
cycle-inducing edge back to the $\GDR$ node, thus completing the definition of a
recursive function. The GDR compilation rule is summarised as
\cref{alg:domainrecursion} and explained in more detail using the example below.

\begin{example}
  Let $\phi \coloneqq \{\, c_1, c_2 \,\}$ be the formula from
  \cref{example:first}. While GDR is possible on both domains, we illustrate how
  it works on $\Gamma$. The algorithm iterates over the clauses of $\phi$.
  Suppose \cref{line:forclause} picks $c = c_1$ as the first clause. Then, set
  $V$ is constructed to contain all variables with domain $\Omega = \Gamma$ that
  occur in the literals of clause $c$. In this case, $V = \{\, X \,\}$.

  \Cref{line:conditions} iterates over all subsets $W \subseteq V$ of variables
  that can be replaced by a constant without resulting in evidently
  unsatisfiable formulas. We impose two restrictions on $W$. First,
  $W^2 \cap E = \emptyset$ ensures that no pairs of variables in $W$ are
  constrained to be distinct since that would result in an $x \ne x$ constraint
  after substitution. Similarly, we want to avoid variables in $W$ that have
  inequality constraints with constants. In this case, both subsets of $V$
  satisfy the conditions, and \cref{line:generation} generates two clauses:
  \begin{multline*}
    (\{\, \neg \predicate(X, Y), \neg \predicate(X, Z) \,\}, \{\, (Y, Z), (X, x) \,\}, \\
    \{\, X \mapsto \Gamma, Y \mapsto \Delta, Z \mapsto \Delta \,\}),
  \end{multline*}
  from $W = \emptyset$ and
  \[
    (\{\, \neg \predicate(x, Y), \neg \predicate(x, Z) \,\}, \{\, (Y, Z) \,\}, \{\, Y \mapsto \Delta, Z \mapsto \Delta \,\})
  \]
  from $W = V$.

  When \cref{line:forclause} picks $c = c_2$, we have $V = \{\, X, Z \,\}$. The
  subset $W = V$ fails to satisfy the conditions on \cref{line:conditions}
  because of the $X \ne Z$ constraint. The other three subsets of $V$ all
  generate clauses for $\phi'$. Indeed, $W = \emptyset$ generates
  \begin{multline*}
    (\{\, \neg \predicate(X, Y), \neg \predicate(Z, Y) \,\}, \{\, (X, Z), (X, x), (Z, x) \,\}, \\
    \{\, X \mapsto \Gamma, Y \mapsto \Delta, Z \mapsto \Gamma \,\}),
  \end{multline*}
  $W = \{\, X \,\}$ generates
  \[
    (\{\, \neg \predicate(x, Y), \neg \predicate(Z, Y) \,\}, \{\, (Z, x) \,\}, \{\, Y \mapsto \Delta, Z \mapsto \Gamma \,\}),
  \]
  and $W = \{\, Z \,\}$ generates
  \[
    (\{\, \neg \predicate(X, Y), \neg \predicate(x, Y) \,\}, \{\, (X, x) \,\}, \{\, X \mapsto \Gamma, Y \mapsto \Delta \,\}).
  \]
\end{example}

\begin{theorem}[Correctness of $\GDR$]\label{thm:correctness1}
  Let $\phi$ be the formula used as input to \cref{alg:domainrecursion},
  $\Omega \in \mathcal{D}$ the domain selected on \cref{line:condition}, and
  $\phi'$ the formula constructed by the algorithm for $\Omega$. Suppose that
  $\Omega \ne \emptyset$. Then $\phi \equiv \phi'$.
\end{theorem}

\Cref{thm:correctness1} is equivalent to Proposition~3 by Van den
Broeck~\shortcite{DBLP:conf/nips/Broeck11}. For the proofs of this and other
theorems, see \cref{sec:proofs}.

\subsection{Constraint Removal}\label{sec:cr}

Recall that GDR on a domain $\Omega$ creates constraints of the form $X_i \ne x$
for some constant $x \in \Omega$ and family of variables $X_i \in \Omega$. Once
some conditions are satisfied, we can eliminate these constraints and replace
$\Omega$ with a new domain $\Omega' \coloneqq \Omega \setminus \{\, x \,\}$.
These conditions are that a constraint of the form $X \ne x$ exists for all
variables $X \in \Omega$ across all clauses, and such constraints are the only
place where $x$ occurs. We formalise the conditions as \cref{def:replaceable}.

\begin{definition}\label{def:replaceable}
  For a formula $\phi$, a pair $(\Omega, x)$ of a domain
  $\Omega \in \mathcal{D}$ and its element $x \in \Omega$ is called
  \emph{replaceable} if
  \begin{enumerate*}[label=(\roman*)]
    \item $x$ does not occur in any literal of any clause of $\phi$, and
    \item for each clause $c = (L, E, \delta_c) \in \phi$ and variable
    $X \in \Vars(c)$, either $\delta_c(X) \ne \Omega$ or $(X, x) \in E$.
  \end{enumerate*}
\end{definition}

\begin{algorithm}[t]
  \caption{The compilation rule for $\CR$ nodes}\label{alg:constraintremoval}
  \KwIn{formula $\phi$, set of all relevant domains $\mathcal{D}$}
  \KwOut{set of chips $S$}
  $S \gets \emptyset$\;
  \ForEach{replaceable pair $(\Omega \in \mathcal{D}, x \in \Omega)$\label{line:crconditions}}{
    add a new domain $\Omega'$ to $\mathcal{D}$\;\label{line:newdomain}
    $\phi' \gets \emptyset$\;
    \ForEach{clause $(L, E, \delta) \in \phi$}{
      $E' \gets \{\, (a, b) \in E \mid b \ne x \,\}$\;\label{line:constraintremoval}
      \nosemic$\delta' \gets X \mapsto
      \begin{cases}
        \Omega' & \text{if } \delta(X) = \Omega\\
        \delta(X) & \text{otherwise;}
      \end{cases}$\;\label{line:newdelta}
      $\phi' \gets \phi' \cup \{\, (L, E', \delta') \,\}$\;
    }
    $S \gets S \cup \{\, (\CR_{\Omega \mapsto \Omega'}, \langle\phi'\rangle) \,\}$\;
  }
\end{algorithm}

Once a replaceable pair is found, \cref{alg:constraintremoval} constructs the
new formula by removing constraints and defining a new domain map $\delta'$ that
replaces $\Omega$ with $\Omega'$.

\begin{example}
  Let $\phi = \{\, c_1, c_2 \,\}$ be a formula with clauses
  \begin{align*}
    c_1 &=
          \begin{multlined}[t]
            (\{\, \neg \predicate(X, Y), \neg \predicate(X, Z) \,\}, \{\, (X, x), (Y, Z) \,\}, \\
            \{\, X \mapsto \Gamma, Y \mapsto \Delta, Z \mapsto \Delta \,\}),
          \end{multlined}\\
    c_2 &=
          \begin{multlined}[t]
            (\{\, \neg \predicate(X, Y), \neg \predicate(Z, Y) \,\}, \\
            \{\, (X, x), (Z, X), (Z, x) \,\}, \\
            \{\, X \mapsto \Gamma, Y \mapsto \Delta, Z \mapsto \Gamma \,\}).
          \end{multlined}
  \end{align*}
  Domain $\Gamma$ and its element $x \in \Gamma$ satisfy the preconditions for
  CR\@. The rule introduces a new domain $\Gamma'$ and transforms $\phi$ to
  $\phi' = (c_1', c_2')$, where
  \begin{align*}
    c_1' &=
           \begin{multlined}[t]
             (\{\, \neg \predicate(X, Y), \neg \predicate(X, Z) \,\}, \{\, (Y, Z) \,\}, \\
             \{\, X \mapsto \Gamma', Y \mapsto \Delta, Z \mapsto \Delta \,\}),
           \end{multlined} \\
    c_2' &=
           \begin{multlined}[t]
             (\{\, \neg \predicate(X, Y), \neg \predicate(Z, Y) \,\}, \{\, (Z, X) \,\}, \\
             \{\, X \mapsto \Gamma', Y \mapsto \Delta, Z \mapsto \Gamma' \,\}).
           \end{multlined}
  \end{align*}
\end{example}

\begin{theorem}[Correctness of $\CR$]
  Let $\phi$ be the input formula of \cref{alg:constraintremoval}, $(\Omega, x)$
  a replaceable pair, and $\phi'$ the output formula for when $(\Omega, x)$ is
  selected on \cref{line:crconditions}. Then $\phi \equiv \phi'$, where the
  domain $\Omega'$ introduced on \cref{line:newdomain} is interpreted as
  $\Omega \setminus \{\, x \,\}$.
\end{theorem}

There is no analogue to $\CR$ in previous work on first-order knowledge
compilation. $\CR$ plays a key role by recognising when the constraints of a
formula essentially reduce the size of a domain by one and extracting this
observation into the definition of a new domain. This then allows us to relate
maps between sets of domains to the arguments of a function call.
\Cref{sec:ref,sec:interpret} describe this process.

\subsection{Identifying Opportunities for Recursion}\label{sec:ref}

\subsubsection{Hashing}
We use (integer-valued) hash functions to discard pairs of formulas too
different for recursion. The hash code of a clause $c = (L, E, \delta_{c})$
(denoted by $\# c$) combines the hash codes of the sets of constants and
predicates in $c$, the numbers of positive and negative literals, the number of
inequality constraints $|E|$, and the number of variables $|\Vars(c)|$. The hash
code of a formula $\phi$ combines the hash codes of all its clauses and is
denoted by $\#\phi$.

\subsubsection{Caching}
\textsc{ForcLift} uses a cache to check if a formula is identical to one of the
formulas that have already been fully compiled. To facilitate recursion, we
extend the caching scheme to include formulas encountered before but not fully
compiled yet. Formally, we define a \emph{cache} as a map from integers (e.g.,
hash codes) to sets of pairs of the form $(\phi, v)$, where $\phi$ is a formula,
and $v$ is an FCG node.

\begin{algorithm}[t]
  \caption{The compilation rule for $\Reff$ nodes}\label{alg:trycache}
  \KwIn{formula $\phi$, cache $C$}
  \KwOut{a set of chips}

  \ForEach{formula and node $(\psi, v) \in C(\#\phi)$}{\label{line:selectformula}
    $\rho \gets \identifyRecursion{$\phi$, $\psi$}$\;
    \lIf{$\rho \ne \nulll$}{\Return{$\{\, (\Reff_\rho(v), \langle\rangle) \,\}$}}\label{line:rho}
  }
  \Return{$\emptyset$}\;

  \Fn{\identifyRecursion{formula $\phi$, formula $\psi$, map $\rho = \emptyset$}}{
    \lIf{$|\phi| \ne |\psi|$ {\bf or} $\#\phi \ne \#\psi$}{\Return{\nulll}}
    \lIf{$\phi = \emptyset$}{\Return{$\rho$}}
    \ForEach{clause $c \in \psi$\label{line:for1}}{
      \ForEach{clause $d \in \phi$ s.t. $\#d=\#c$\label{line:for2}}{
        \ForEach{$\gamma \in \generateMaps{$c$, $d$, $\rho$}$\label{line:generateMaps}}{
          $\rho' \gets \identifyRecursion{$\phi\setminus\{\, d \,\}$, $\psi\setminus\{\, c \,\}$, $\rho\cup\gamma$}$\;\label{line:recursion}
          \lIf{$\rho' \ne \nulll$}{\Return{$\rho'$}}
        }
      }
      \Return{\nulll}\;
    }
  }
\end{algorithm}

% Cache $C$ is used to partition all previously-encountered formulas based on
% their hash codes.
\Cref{alg:trycache} describes the compilation rule for creating $\Reff$ nodes.
For every formula $\psi$ in the cache s.t. $\#\psi = \#\phi$, function
\identifyRecursion checks whether a recursive call is feasible. If it is,
\identifyRecursion returns a (total) map
$\rho\colon \Doms(\psi) \to \Doms(\phi)$ that shows how $\psi$ can be
transformed into $\phi$ by replacing each domain $\Omega \in \Doms(\psi)$ with
$\rho(\Omega) \in \Doms(\phi)$. Otherwise, \identifyRecursion returns \nulll{}
to signify that $\phi$ and $\psi$ are too different for recursion to work. This
happens if $\phi$ and $\psi$ (or their subformulas explored in recursive calls)
are structurally different (i.e., the numbers of clauses or the hash codes fail
to match) or if a clause of $\psi$ cannot be paired with a sufficiently similar
clause of $\phi$. Function \identifyRecursion works by iterating over pairs of
clauses of $\phi$ and $\psi$ with the same hash codes. For every pair of similar
clauses, \identifyRecursion calls itself on the remaining clauses until the map
$\rho\colon \Doms(\psi) \pfun \Doms(\phi)$ becomes total.

Function \generateMaps checks the compatibility of a pair of clauses. It
considers every possible bijection $\beta\colon \Vars(c) \to \Vars(d)$ and map
$\gamma\colon \Doms(c) \to \Doms(d)$ s.t.
\[
  \begin{tikzcd}
    \Vars(c) \ar[r, dashed, "\beta"] \arrow[d, swap, "\delta_c"] & \Vars(d) \ar[d, "\delta_d"] \\
    \Doms(c) \ar[r, dashed, "\gamma"] \ar[d, hookrightarrow] & \Doms(d) \ar[d, hookrightarrow] \\
    \Doms(\psi) \ar[r, swap, "\rho", "|" marking, outer sep=5pt] & \Doms(\phi).
  \end{tikzcd}
\]
commutes, and $c$ becomes equal to $d$ when its variables are replaced according
to $\beta$ and its domains replaced according to $\gamma$. The function then
returns each such $\gamma$ as soon as possible.

\begin{theorem}[Correctness of $\Reff$]
  Let $\phi$ be the formula used as input to \cref{alg:trycache}. Let $\psi$ be
  any formula selected on \cref{line:selectformula} of the algorithm s.t.\
  $\rho \ne \nulll$ on \cref{line:rho}. Let $\sigma$ be a domain size function.
  Then the set of models of $(\psi, \sigma \circ \rho)$ is equal to the set of
  models of $(\phi, \sigma)$.
\end{theorem}

As \textsc{ForcLift} only supports $\Reff$ nodes when the two formulas are
equal, our approach is much more general and capable of creating function calls
as complex as $f(n-k-2)$.

\section{Converting FCGs to Function Definitions}\label{sec:interpret}

\begin{algorithm}[t]
  \caption{Construct functions from an FCG}\label{alg:interpretation}
  \KwIn{$\mathcal{D}$---the set of domains of the input formula}
  \KwIn{$s$---the source node of the FCG}
  \KwData{$\functionNames = \emptyset$ (function names)}
  \KwOut{a list of function definitions}
  $\domainNames \gets \{\, \Omega \mapsto \newDomainName{} \mid \Omega \in \mathcal{D} \,\}$\;\label{line:varnames}
  \tcc{Equivalent condition: $s$ has in-degree greater than one.}
  \If{$s$ is a direct successor of a $\Reff$ node}{
    $(\expression, \functions) \gets \visit{$s$, $\domainNames$}$\;\label{line:initialcall}
    \Return{\functions}\;\label{line:initialcalltwo}
  }
  $f \gets \newFunctionName{$s$}$\;
  $(\expression, \functions) \gets \visit{$s$, $\domainNames$}$\;
  \Return{$\langle f(\langle \domainNames{$\Omega$} \mid \Omega \in \mathcal{D} \rangle) = \expression \rangle \mdoubleplus \functions$}\;
  \Fn{\visit{node $v$, domain names \domainNames}}{
    \If{$v$ is \textbf{not} a direct successor of a $\Reff$ node}{
      \Return{\actuallyVisit{$v$, \domainNames}}\;
    }
    $f \gets \newFunctionName{$v$}$\;
    $\newDomainNames \gets \domainNames$\;
    \ForEach{$\Omega \in \mathfrak{D}(v)$ s.t. $\newDomainNames(\Omega)$ is non-atomic}{\label{line:nonone}
      $\newDomainNames(\Omega) \gets \newDomainName{}$\;\label{line:nontwo}
    }
    $(\expression, \functions) \gets \actuallyVisit{$v$, \newDomainNames}$\;\label{line:actuallyvisit}
    $\functionCall \gets f(\langle \domainNames{$\Omega$} \mid \Omega \in \mathfrak{D}(v) \rangle)$\;
    $\functionSignature \gets f(\langle \newDomainNames{$\Omega$} \mid \Omega \in \mathfrak{D}(v) \rangle)$\;
    \Return{$(\functionCall, \langle \functionSignature = \expression \rangle \mdoubleplus \functions)$}\;
  }
  \Fn{\actuallyVisit{node $v$, domain names \domainNames}}{
    \Switch{label of $v$}{
      \lCase{$\GDR(v')$}{\Return{\visit{$v'$, \domainNames}}}
      \uCase{$\CR_{\Omega \mapsto \Omega'}(v')$}{
        \Return{\visit{$v'$, $\domainNames \cup \{\, \Omega' \mapsto (\domainNames(\Omega) - 1) \,\}$}}\;\label{line:cr}
      }
      \uCase{$\Reff_{\rho}(v')$}{
        $\args \gets \langle \domainNames\circ\rho(\Omega) \mid \Omega \in \mathfrak{D}(v') \rangle$\;\label{line:args}
        \Return{$(\functionNames{$v'$}(\args), \langle \rangle)$}\;
      }
      \dots
    }
  }
\end{algorithm}

\Cref{alg:interpretation} constructs a list of function definitions from an
FCG\@. The algorithm consists of two main functions: \visit and \actuallyVisit.
The former handles new function definitions, while the latter produces an
algebraic interpretation of each node depending on its type. As there are many
node types, we only include the ones pertinent to the contributions of this
paper; see previous work~\cite{DBLP:conf/ijcai/BroeckTMDR11} for information
about other types. Given a node $v$ as input, both \visit and \actuallyVisit
return a pair $(\expression, \functions)$. Here, \expression is the algebraic
expression representing $v$, and \functions is a list of auxiliary functions
created while formulating \expression.

The algorithm gradually constructs two partial maps \functionNames and
\domainNames providing names to functions and domains, respectively.
\functionNames is a global variable that maps FCG nodes to function names (e.g.,
$f$, $g$). \domainNames maps domains to their names, and it is passed as an
argument to \visit and \actuallyVisit. Here, a \emph{domain name} is either a
parameter of a function or an algebraic expression consisting of function
parameters, subtraction, and parentheses. We call a domain name \emph{atomic} if
it is free of subtraction (e.g., $m$, $n$); otherwise it is \emph{non-atomic}
(e.g., $m-1$, $n-l$). Functions \newDomainName and \newFunctionName both
generate previously-unused names. The latter also takes an FCG node as input and
links it with the new name in \functionNames.

We assume a fixed total ordering of all domains. In particular, in
\cref{example:interpretation} below, let $\Gamma$ go before $\Delta$. For each
node $v$, let $\mathfrak{D}(v)$ denote the (pre-computed) list of domains, the
sizes of which we would use as parameters if we were to define a function that
begins at $v$. As a set, it is computed by iteratively setting
$\mathfrak{D}(v) \gets \left(\bigcup_{u} \mathfrak{D}(u) \setminus I_{v}\right) \cup U_{v}$
until convergence, where:
\begin{enumerate*}[label=(\roman*)]
  \item the union is over the direct successors of $v$,
  \item $I_{v}$ is the set of domains \emph{introduced} at node $v$, and
  \item $U_{v}$ is the set of domains \emph{used} by $v$.
\end{enumerate*}
The set is then sorted according to the ordering.

\begin{figure}
  \centering
  \begin{tikzcd}
    \ar[d] & \\
    \text{\Cref{alg:interpretation}} \ar[d] \ar[u, swap, shift right=0.7ex, "{\left\langle f(m, n) = \sum_{l=0}^{n}[l<2]f(m-1,n-l) \right\rangle}"] & \\
    \visit{$\GDR$, $\Done$} \ar[d] \ar[u, swap, shift right=0.7ex, "\left({f(m, n), \left\langle f(m, n) = \sum_{l=0}^{n}[l<2]f(m-1,n-l) \right\rangle}\right)"] & \\
    \actuallyVisit{$\GDR$, $\Done$} \ar[d] \ar[u, swap, shift right=0.7ex, "\left({\sum_{l=0}^{n}[l<2]f(m-1,n-l), \langle\rangle}\right)"] & \\
    \visit/\actuallyVisit{$\bigvee$, $\Done$} \ar[d] \ar[u, swap, shift right=0.7ex, "\left({\sum_{l=0}^{n}[l<2]f(m-1,n-l), \langle\rangle}\right)"] & \\
    \visit/\actuallyVisit{$\CR$, $\Dtwo$} \ar[d] \ar[u, swap, shift right=0.7ex, "\left({[l<2]f(m-1,n-l), \langle\rangle}\right)"] & \\
    \visit/\actuallyVisit{$\land$, $\Dthree$} \ar[d] \ar[dr] \ar[u, swap, shift right=0.7ex, "\left({[l<2]f(m-1,n-l), \langle\rangle}\right)"] & \\
    \visit/\actuallyVisit{$\bot$, $\Dthree$} \ar[u, swap, shift right=0.7ex, "\left({[l<2], \langle\rangle}\right)"] & \visit/\actuallyVisit{$\Reff$, $\Dthree$} \ar[ul, swap, shift right=0.7ex, "\left({f(m-1, n-l), \langle\rangle}\right)"]
  \end{tikzcd}
  \caption{Function calls and their return values when \cref{alg:interpretation}
    is run on the FCG from \cref{fig:examplefcg}}\label{fig:functioncalls}
\end{figure}
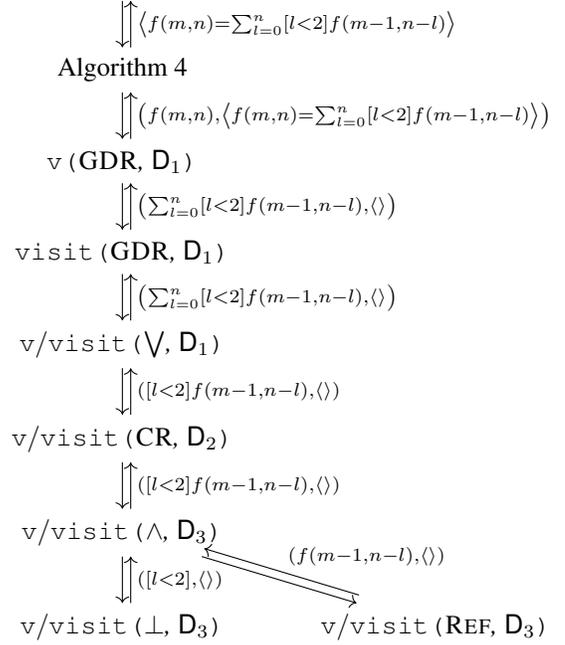

\begin{example}\label{example:interpretation}
  Here we examine how \cref{alg:interpretation} works on the FCG from
  \cref{fig:examplefcg}. In this case, $\mathcal{D} = \{\, \Gamma, \Delta \,\}$.
  Let $\Done \coloneqq \{\, \Gamma \mapsto m, \Delta \mapsto n \,\}$,
  $\Dtwo \coloneqq \Done \cup \{\, \Delta^{\top} \mapsto l, \Delta^{\bot} \mapsto (n-l) \,\}$,
  and $\Dthree \coloneqq \Dtwo \cup \{\, \Gamma' \mapsto (m-1) \,\}$ denote
  versions of \domainNames at various points throughout the algorithm's
  execution. Here, $l$, $m$, and $n$ are all arbitrary names generated by
  \newDomainName.

  \Cref{fig:functioncalls} shows the outline of function calls and their return
  values. For simplicity, we refer to FCG nodes by their types. When \visit
  calls \actuallyVisit and returns its output unchanged, we shorten the two
  function calls to $\visit/\actuallyVisit$.

  \Cref{line:varnames} initialises \domainNames to $\Done$. Once \visit{$\GDR$,
    $\Done$} is called on \cref{line:initialcall}, \newFunctionName updates
  \functionNames to $\{\, \GDR \mapsto f \,\}$, where $f$ is a new function
  name. There are no non-atomic names in $\Done$, so
  \cref{line:nonone,line:nontwo} are skipped, and calls to
  \actuallyVisit{$\GDR$, $\Done$} and
  $\visit/\actuallyVisit{$\bigvee$, $\Done$}$ follow. Here, $\Done$ becomes
  $\Dtwo$, i.e., the $\bigvee$ node introduces two domains $\Delta^{\top}$ and
  $\Delta^{\bot}$ that `partition' $\Delta$, i.e.,
  $\Dtwo(\Delta^{\top}) + \Dtwo(\Delta^{\bot}) = l + (n-l) = n = \Dtwo(\Delta)$.
  Similarly, \actuallyVisit{$\CR$, $\Dtwo$} on \cref{line:cr} adds $\Gamma'$,
  replacing $\Dtwo$ with $\Dthree$.

  Eventually, \actuallyVisit{$\bot$, $\Dthree$} is called. A contradiction node
  with clause $c$ as a parameter is interpreted as one if the clause has
  groundings and zero otherwise. In this case, the parameter is
  $c = (\emptyset, \{\, (X, Y) \,\}, \{\, X \mapsto \Delta^\top, Y \mapsto \Delta^\top \,\})$
  (not shown in \cref{fig:examplefcg}). Equivalently,
  $\forall X, Y \in \Delta^{\top}\text{. }X \ne Y \Rightarrow \bot$, i.e.,
  $\forall X, Y \in \Delta^{\top}\text{. }X = Y$, i.e.,
  $\Dthree(\Delta^{\top}) < 2$. We use the Iverson bracket notation to write
  \[
    [l < 2] \coloneqq
    \begin{cases}
      1 & \text{if } l < 2 \\
      0 & \text{otherwise}
    \end{cases}
  \]
  for the output expression of \actuallyVisit{$\bot$, $\Dthree$}.

  Next, \actuallyVisit{$\land$, $\Dthree$} calls
  $\visit/\actuallyVisit{$\Reff$, $\Dthree$}$. Since
  $\mathfrak{D}(\GDR) = \langle \Gamma, \Delta \rangle$, and the pair
  $\langle\Gamma, \Delta\rangle$ is transformed by $\rho$ to
  $\langle\Gamma', \Delta^{\bot}\rangle$ and then by $\Dthree{}$ to
  $\langle m-1, n-l \rangle$, \args is set to $\langle m-1, n-l \rangle$ on
  \cref{line:args}, and the output expression becomes $f(m-1, n-l)$. The call to
  \actuallyVisit{$\land$, $\Dthree$} then returns the product of the two
  expressions $[l<2]f(m-1, n-l)$. Next, \actuallyVisit{$\bigvee$, $\Done$}
  returns $\sum_{l=0}^{n}[l<2]f(m-1,n-l)$. The same expression (and an empty
  list of functions) is then returned to \cref{line:actuallyvisit} of the call
  to \visit{$\GDR$, $\Done$}. Here, both \functionCall and \functionSignature
  are set to $f(m, n)$. As the definition of $f$ is the final answer and not
  part of some other algebraic expression,
  \cref{line:initialcall,line:initialcalltwo} discard the function call
  expression and return the definition of $f$. Thus,
  \begin{align}
    f(m, n) &= \sum_{l = 0}^{n} \binom{n}{l} [l < 2] f(m-1, n-l)\nonumber \\
            &= f(m-1, n) + n f(m-1, n-1)\label{eq:solution}
  \end{align}
  is the function that \textsc{Crane} constructed for computing partial
  injections. To use $f$ in practice, one has to identify the base cases
  $f(0, n)$ and $f(m, 0)$ for all $m, n \in \mathbb{N}_{0}$.
\end{example}

\section{Complexity Results}\label{sec:results}

\begin{table*}[t]
  \centering
  \begin{tabular}{cccccc}
    \toprule
    \multicolumn{3}{c}{Function class} & \multicolumn{3}{c}{Complexity of $\mathcal{F}_{\phi}$ (as in \cref{fig:process})} \\
    Partial & Endo- & Class & By Hand & With \textsc{ForcLift} & With \textsc{Crane} \\
    \midrule
    \rowcolor{gray!25}\cmark/\xmark & \cmark/\xmark & Functions & $\log m$ & $m$ & $m$ \\
    \xmark & \xmark & \multirow{2}{*}{Surjections} & $n \log m$ & $m^{3}+n^{3}$ & $m^{3}+n^{3}$ \\
    \xmark & \cmark & & $m \log m$ & $m^{3}$ & $m^{3}$ \\
    \rowcolor{gray!25}\xmark & \xmark & & $m$ & --- & $mn$ \\
    \rowcolor{gray!25}\xmark & \cmark & & $m$ & --- & $m^3$ \\
    \rowcolor{gray!25}\cmark & \xmark & & ${\min\{\, m, n \,\}}^2$ & --- & $mn$ \\
    \rowcolor{gray!25}\cmark & \cmark & \multirow{-4}{*}{Injections} & $m^2$ & --- & --- \\
    \xmark & \xmark & Bijections & $m$ & --- & $m$ \\
    \bottomrule
  \end{tabular}
  \caption{The worst-case complexity of counting various types of functions. All
    asymptotic complexities are in $\Theta(\cdot)$. A dash means that the
    algorithm was not able to find a solution. In the case of \textsc{ForcLift},
    this means that the greedy search algorithm ended with a formula unsuitable
    for any compilation rule. In the case of \textsc{Crane}, this means that a
    complete solution could not be found after having explored the maximum
    allowed depth of the search tree.}\label{tbl:results}
\end{table*}

We compare \textsc{Crane} and
\textsc{ForcLift}\footnote{\url{https://dtaid.cs.kuleuven.be/wfomc}} on their
ability to count various classes of functions. We chose this class of instances
because of its simplicity and the inability of publicly available WFOMC
algorithms to solve many such counting problems. Note that, except for a
particular version of \textsc{FastWFOMC}, other exact WFOMC algorithms cannot
solve any instances \textsc{ForcLift} fails on. First, we describe how to
express these function-counting problems in first-order logic. \textsc{ForcLift}
then translates these sentences in first-order logic to formulas as in
\cref{def:formula}.

Let $\predicate$ be a predicate that models relations between sets $\Gamma$ and
$\Delta$. To restrict all such relations to just $\Gamma \to \Delta$ functions,
one might write $\forall X \in \Gamma\text{. }\forall Y,Z \in \Delta\text{.
}\predicate(X, Y) \land \predicate(X, Z) \Rightarrow Y = Z$ and
\begin{equation}\label{eq:def2}
  \forall X \in \Gamma\text{. }\exists Y \in \Delta\text{. }\predicate(X, Y).
\end{equation}
The former sentence says that one element of $\Gamma$ can map to at \emph{most}
one element of $\Delta$, and the latter sentence says that each element of
$\Gamma$ must map to at \emph{least} one element of $\Delta$. One can then add
$\forall X,Z \in \Gamma\text{. }\forall Y \in \Delta\text{.
}\predicate(X, Y) \land \predicate(Z, Y) \Rightarrow X = Z$ to restrict
$\predicate$ to injections or $\forall Y \in \Delta\text{.
}\exists X \in \Gamma\text{. }\predicate(X, Y)$ to ensure surjectivity or remove
\cref{eq:def2} to consider partial functions. Lastly, one can replace all
occurrences of $\Delta$ with $\Gamma$ to model endofunctions (i.e., functions
with the same domain and codomain) instead.

% the process: both for running algorithms and for determining complexities
We consider all sixteen combinations of these properties (injectivity,
surjectivity, partiality, and endo-), omitting duplicate descriptions of the
same function, e.g., the number of partial $\Gamma \to \Delta$ bijections is the
same as the number of $\Delta \to \Gamma$ injections. We run each algorithm once
on each instance. \textsc{Crane} is run in hybrid search mode until either it
finds five solutions or the search tree reaches height six. \textsc{ForcLift} is
always run until it terminates. If successful, \textsc{Crane} returns one or
more sets of function definitions, and \textsc{ForcLift} returns a circuit,
which can similarly be interpreted as a function definition. We then assess the
complexity of each solution by hand and pick the best in case \textsc{Crane}
returns several solutions of varying complexities. In our experiments,
\textsc{Crane} found at most four solutions per problem instance, and most
solution had the same complexity.

% Automatically determining the complexity of a solution generated by
% Crane is another area for future work. Alternatively, one could evaluate
% several solutions in parallel until one of them produces the model count. In
% our experiments, Crane found at most four solutions per problem instance. Most
% had the same complexity, although some were less efficient (e.g., cubic
% instead of linear or quintic instead of cubic).

% how the complexity of a solution is determined

Recall that, in the case of \textsc{ForcLift}, solving a (W)FOMC problem
consists of two parts: compilation (via search) and propagation. (A similar
dichotomy applies to \textsc{Crane} as well.) The complexity of the former
depends only on the formula and so far has received very little attention except
for the work by Beame et al.~\shortcite{DBLP:conf/pods/BeameBGS15}. The
complexity of the latter, on the other hand, is a function of domain sizes and
can be determined without measuring runtime. We establish the asymptotic
complexity of a solution by counting the number of arithmetic operations needed
to follow the definitions of constituent functions without recomputing the same
quantity multiple times.\footnote{Although this is not done by default, one
  could also use arbitrary-precision arithmetic and consider bit complexity
  instead.} In particular, we assume that each function call and binomial
coefficient is computed at most once, and computing $\binom{n}{k}$ takes
$\Theta(nk)$ time. For example, the complexity of \cref{eq:solution} is
$\Theta(mn)$ since $f(m, n)$ can be computed by a dynamic programming algorithm
that computes $f(i, j)$ for all $i = 0, \dots, m$ and $j = 0, \dots, n$, taking
a constant amount of time on each $f(i, j)$.

Let $m = |\Gamma|$ and $n = |\Delta|$ be domain sizes. We summarise the results
in \cref{tbl:results}, where we compare the solutions found by both algorithms
to the manually-constructed ways of computing the same quantities that have been
proposed before.\footnote{\url{https://oeis.org},
  \url{https://scicomp.stackexchange.com/q/30049}} Examining the last two
columns of the table, we see that the two algorithms perform equally well on
instances that could already be solved by \textsc{ForcLift}. However,
\textsc{Crane} can also solve all but one of the instances that
\textsc{ForcLift} fails on in at most cubic time. See \cref{sec:solutions} for
the exact solutions produced by \textsc{Crane}.

\subsection{A Comparison with \textsc{\normalfont FastWFOMC}}

The function-counting problems described above can be expressed in \FOtwo{} with
cardinality constraints, and so are known to be
liftable~\cite{DBLP:journals/jair/Kuzelka21}. However, two algorithms that both
run in polynomial time are not necessarily equally good: the degree of the
polynomial can make a substantial difference. Hence, we compare \textsc{Crane}
with \textsc{FastWFOMC}~\cite{DBLP:conf/uai/BremenK21}, extended to support
cardinality constraints, which was provided to us by one of the authors. We
compare them on the task of counting permutations of a set, i.e.,
bijective/injective endofunctions. We can describe this problem in \FOtwo{} with
cardinality constraints as
\begin{align*}
  (|\predicate| = |\Delta|) \land{} &(\forall X, Y \in \Delta\text{. } \predicater(X) \lor \neg \predicate(X, Y)) \land{}\\
  &(\forall X, Y \in \Delta\text{. } \predicates(X) \lor \neg\predicate(Y, X)),
\end{align*}
where $\predicater$ and $\predicates$ are Skolem predicates as described by Van
den Broeck et al.~\shortcite{DBLP:conf/kr/BroeckMD14}.

\begin{figure}[t]
  \includegraphics{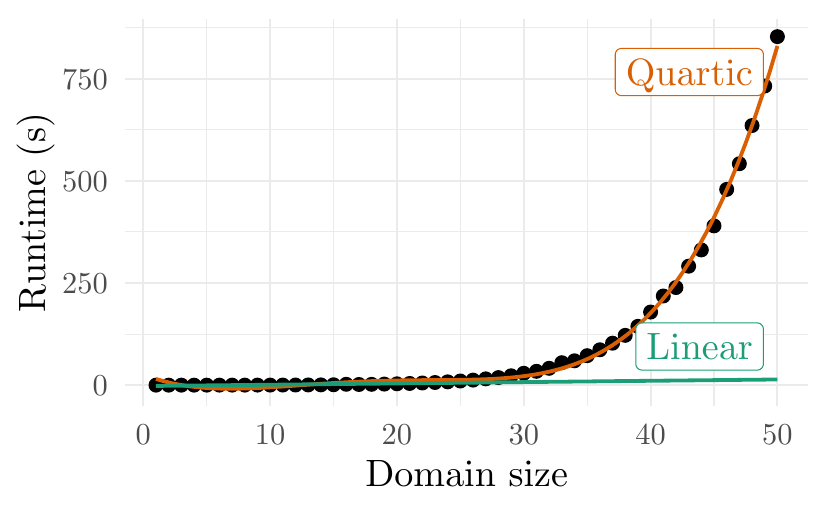}
  \caption{The runtime of \textsc{FastWFOMC} when counting permutations
    (indicated by points), together with a quartic regression model fitted on
    all of the data and a linear model fitted on the first half of the data
    points}\label{fig:fastwfomc}
\end{figure}

We run \textsc{FastWFOMC} on an AMD Ryzen~7~5800H processor with
\SI{16}{\giga\byte} of memory, running Arch Linux~6.2.2-arch1-1 and
Python~3.10.9, with $|\Delta|$ ranging from 1 to 50. The results are in
\cref{fig:fastwfomc}. We use ten-fold cross-validation to check the degree of
the polynomial that best fits the data. The best-fitting degree is 6 with an
average mean squared error (aMSE) of 16.9, although all degrees above 4 fit the
data almost equally well. As \cref{fig:fastwfomc} demonstrates, a quartic (i.e.,
degree 4) polynomial (with an aMSE of 113.5) fits well too. However, the aMSE
quickly grows to 949.5 and higher values for degrees less than 4. According to
\cref{tbl:results}, \textsc{Crane} finds a cubic solution to this problem.
However, we can also use the linear solution for counting bijections between
different domains for this task by setting $n \coloneqq m$.

\subsection{Miscellaneous Other Instances}

We also note that standard (W)FOMC benchmarks such as
\cref{example:smokers}---that are already supported by
\textsc{ForcLift}---remain feasible for \textsc{Crane} as well. To demonstrate
that \textsc{Crane} is capable of finding efficient solutions to problems beyond
known domain-liftable fragments such as \Ctwo{}, we run \textsc{Crane} in greedy
search mode on the following formula:
\begin{equation}\label[formula]{eq:ternary}
  \begin{split}
    &\begin{multlined}[t]
      (\forall X,W \in \Gamma\text{. }\forall Y \in \Delta\text{. }\forall Z \in \Lambda\text{. }\\
      \predicate(X, Y, Z) \land \predicate(W, Y, Z) \Rightarrow X = W) \land{}
    \end{multlined}\\
    &\begin{multlined}[t]
      (\forall X \in \Gamma\text{. }\forall Y,W \in \Delta\text{. }\forall Z \in \Lambda\text{. }\\
      \predicate(X, Y, Z) \land \predicate(X, W, Z) \Rightarrow Y = W).
    \end{multlined}
  \end{split}
\end{equation}
For every element $z$ of $\Lambda$, \cref{eq:ternary} restricts
$\predicate(X, Y, z)$ to be a $\Gamma \to \Delta$ partial injection, independent
of $\predicate(X, Y, z')$ for any $z' \ne z$. \textsc{Crane} instantly returns a
solution of complexity $\Theta(l+mn)$, where $l = |\Lambda|$.

\section{Discussion}\label{sec:conclusion}
\setlength{\multlinegap}{0pt}

This paper presents \textsc{Crane}---an extension of
\textsc{ForcLift}~\cite{DBLP:conf/ijcai/BroeckTMDR11} that benefits from a more
general version of domain recursion and support for graphs with cycles. In
\cref{sec:results}, we listed a range of counting problems that became newly
feasible as a result, including instances outside of currently-known
domain-liftable fragments such as \Ctwo{}. (Although we focused on unweighted
counting, \textsc{ForcLift}'s support for weights trivially transfers to
\textsc{Crane} too.) The common thread across these newly liftable problems is a
version of partial injectivity. Thus, we formulate the following conjecture for
a class of formulas with an injectivity-like constraint on at most two
parameters of a predicate $\predicate$.

\begin{conjecture}\label{conj}
  Let \IFO{} be the class of formulas in first-order logic that contain clauses
  with at most two variables as well as any number of copies of
  \begin{align*}
    &\begin{multlined}[t][0.9\linewidth]
      (\forall X_{1} \in \Delta_{1}\text{. }\dots\,\forall X_{n} \in \Delta_{n}\text{. }\forall Y \in \Delta_{i}\text{. }\predicate(X_{1}, \dots, X_{n}) \land{}\\
      \predicate(X_{1}, \dots, X_{i-1}, Y, X_{i+1}, \dots, X_{n}) \Rightarrow X_{i} = Y) \land{}
    \end{multlined}\\
    &\begin{multlined}[t][0.9\linewidth]
      (\forall X_{1} \in \Delta_{1}\text{. }\dots\,\forall X_{n} \in \Delta_{n}\text{. }\forall Y \in \Delta_{j}\text{. }\predicate(X_{1}, \dots, X_{n}) \land{}\\
      \predicate(X_{1}, \dots, X_{j-1}, Y, X_{j+1}, \dots, X_{n}) \Rightarrow X_{j} = Y)
    \end{multlined}
  \end{align*}
  for some predicate $\predicate/n$, domains
  $\Delta_{1}, \dots, \Delta_{n}$, and $i, j \in \{\, 1, \dots, n \,\}$. Then
  \IFO{} is liftable by \textsc{Crane}.
\end{conjecture}

Note that the unsolved instance in \cref{tbl:results} does not invalidate
\cref{conj} for two reasons. First, in our experiments, the depth of the search
procedure is limited, so there is a possibility that increasing the depth would
make the algorithm succeed. Second, the conjecture assumes distinct domains,
while the unsolved instance pertains to endofunctions.

% In general, our results demonstrate that Crane
% performs better when all domains are distinct. A solution for distinct domains
% can always be applied to the case when some domains are equal.

% ===========================================

% Since most of the instances in \cref{sec:results} are in \Ctwo{}, we can also
% conjecture that \Ctwo{} is liftable by \textsc{Crane}.

% Every sentence in C2 can be rewritten to avoid counting quantifiers by
% introducing more variables. For example, to claim that there are exactly two
% Xs that satisfy p(X), one would introduce two existentially-quantified
% variables X_1 and X_2 such that X_1 != X_2, p(X_1), and p(X_2) and add a
% constraint that any other variable X_3 such that p(X_3) has to be equal to
% either X_1 or X_2.

The most important direction for future work is to automate the process in
\cref{fig:processcrane}. First, we need a way to find the base cases for the
recursive definitions produced by \textsc{Crane}. Second, since the first
solution found by \textsc{Crane} is not always optimal in terms of its
complexity, an automated way to determine the asymptotic complexity of a
solution would be helpful as well. Third, as we saw at the end of
\cref{example:interpretation}, the functions constructed by \textsc{Crane} can
often benefit from elementary algebraic simplifications. This process can be
automated using a computer algebra system. Furthermore, note that having the
solution to a (W)FOMC problem expressed in terms of functions enables many new
possibilities such as
\begin{enumerate*}[label=(\roman*)]
  \item the use of more sophisticated simplification techniques,
  \item asymptotic analysis of, e.g., how the model count grows as the domain
  size goes to infinity, and
  \item answering questions parameterised by domain sizes, e.g., `how big does
  domain $\Delta$ have to be for the probability of event $E$ to be above
  \SI{95}{\percent}?'
\end{enumerate*}

\section*{Acknowledgments}

We thank Anna L.D. Latour and Kuldeep S. Meel for comments on earlier drafts of
this paper. Most of the work for this paper was done while the first author was
a PhD student at the University of Edinburgh. The first author was supported by
the EPSRC Centre for Doctoral Training in Robotics and Autonomous Systems,
funded by the UK Engineering and Physical Sciences Research Council (grant
EP/L016834/1). The second author was partly supported by a Royal Society
University Research Fellowship, UK, and partly supported by a Grant from the
UKRI Strategic Priorities Fund, UK to the UKRI Research Node on Trustworthy
Autonomous Systems Governance and Regulation (EP/ V026607/1, 2020–2024). For the
purpose of open access, the authors have applied a Creative Commons Attribution
(CC BY) licence to any Author Accepted Manuscript version arising from this
submission.

\appendix
\section{Proofs}\label{sec:proofs}

We begin with a few definitions that formalise the notion of a model. For any
set $S$, let $2^{S}$ denote its power set and
$S^{\ast} \coloneqq \bigcup_{i=0}^{\infty} S^{i}$ the (infinite) set of
tuples---of any finite length---of elements of $S$. For example,
${\{\,a, b\,\}}^{\ast} = \{\, (), (a), (b), (a, a), \dots \,\}$. For any
$n \in \mathbb{N}_{0}$, let $[n] \coloneqq \{\, 1, \dots, n \,\}$, e.g.,
$[0] = \emptyset$, and $[2] = \{\, 1, 2 \,\}$. Extending the notation for
substitution introduced in \cref{sec:recprelims}, for a formula $\phi$, constant
$c$, and set of variables $V$, let $\phi[c/V]$ denote $\phi$ with all
occurrences of all variables in $V$ replaced with $c$. Let $\Preds$ be the
function that maps any clause or formula to the set of predicates used in it.
With each formula $\phi$, we associate a map
$\pi_{\phi}\colon \Preds(\phi) \to {\Doms(\phi)}^{\ast}$ s.t., for each
predicate $\predicate/n \in \Preds(\phi)$, we have that
$\pi_{\phi}(\predicate) \in {\Doms(\phi)}^{n}$. This map makes explicit the idea
that each $n$-ary predicate is associated with an $n$-tuple of domains. In the
case of formula $\phi$ from \cref{example:first},
$\pi_{\phi} = \{\, \predicate \mapsto (\Gamma, \Delta) \,\}$.

\begin{definition}\label{def:model}
  Let $\phi$ be a formula and $\sigma$ a domain size function. A \emph{model} of
  $(\phi, \sigma)$ is a map
  $\mathfrak{M}\colon \Preds(\phi) \to 2^{\mathbb{N}_{0}^{\ast}}$ s.t.\ the
  following two conditions are satisfied.
  \begin{enumerate}
    \item For all predicates $\predicate/n \in \Preds(\phi)$ with
          $\pi_{\phi}(\predicate) \coloneqq {(\Delta_{i})}_{i=1}^{n}$ for some
          domains $\Delta_{i} \in \Doms(\phi)$,
    \begin{equation}\label{eq:subset}
      \mathfrak{M}(\predicate) \subseteq \prod_{i=1}^{n}[\sigma(\Delta_{i})].\footnote{For simplicity, \cref{def:model} ignores constants. To include constants, one would replace $[\sigma(\Delta_{i})]$ with a set that contains all constants associated with domain $\Delta_{i}$, extended with enough new elements to make its cardinality $\sigma(\Delta_{i})$.}
    \end{equation}
    \item $\mathfrak{M} \models \phi$ according to the (non-exhaustive)
          recursive definition below. Let $\chi$ and $\psi$ be formulas,
          $\Delta$ a domain, $c \in \Delta$ a constant, $\predicate/n$ a
          predicate associated with domains $\Delta_{1},\dots,\Delta_{n}$, and
          $t \in \prod_{i=1}^{n} \Delta_{i}$ a tuple of constants.
    \begin{itemize}
      \item $\mathfrak{M} \models \predicate(t)$ if and only if (iff)
            $t \in \mathfrak{M}(\predicate)$.
      \item $\mathfrak{M} \models \neg\chi$ iff $\mathfrak{M} \not\models \chi$.
      \item $\mathfrak{M} \models \chi \land \psi$ iff
            $\mathfrak{M} \models \chi$ and $\mathfrak{M} \models \psi$.
      \item $\mathfrak{M} \models \chi \lor \psi$ iff
            $\mathfrak{M} \models \chi$ or $\mathfrak{M} \models \psi$.
      \item $\mathfrak{M} \models (\forall X \in \Delta\text{. }\chi)$ iff
            $\mathfrak{M} \models \chi[x/\{\, X \,\}]$ for all constants
            $x \in \Delta$.
      \item $\mathfrak{M} \models (\forall X \in \Delta\text{.
            } X \ne c \Rightarrow \chi)$ iff
            $\mathfrak{M} \models \chi[x/\{\, X \,\}]$ for all
            $x \in \Delta \setminus \{\, c \,\}$.
      \item $\mathfrak{M} \models (\forall X,Y \in \Delta\text{.
            } X \ne Y \Rightarrow \chi)$ iff
            $\mathfrak{M} \models \chi[x/\{\, X \,\}][y/\{\, Y \,\}]$ for all
            $x,y \in \Delta$ s.t. $x \ne y$.
    \end{itemize}
  \end{enumerate}
\end{definition}

\begin{theorem}[Correctness of $\GDR$]
  Let $\phi$ be the formula used as input to \cref{alg:domainrecursion},
  $\Omega \in \mathcal{D}$ the domain selected on \cref{line:condition}, and
  $\phi'$ the formula constructed by the algorithm for $\Omega$. Suppose that
  $\Omega \ne \emptyset$. Then $\phi \equiv \phi'$.
\end{theorem}
\begin{proof}
  Let $x$ be the constant introduced on \cref{line:constant} and $c$ a clause of
  $\phi$ selected on \cref{line:forclause} of the algorithm. We will show that
  $c$ is equivalent to the conjunction of clauses added to $\phi'$ on
  \cref{line:conditions,line:generation}.

  If $c$ has no variables with domain $\Omega$,
  \cref{line:conditions,line:generation} add a copy of $c$ to $\phi'$. If $c$
  has one variable with domain $\Omega$, say, $X$, then
  $c \equiv \forall X \in \Omega\text{. } \psi$ for some formula $\psi$ with one
  free variable $X$. In this case, \cref{line:conditions,line:generation} create
  two clauses: $\psi[x/\{\, X \,\}]$ and $\forall X \in \Omega\text{.
  } X \ne x \Rightarrow \psi$. It is easy to see that their conjunction is
  equivalent to $c$ provided that $\Omega \ne \emptyset$ and so $x$ is a
  well-defined constant.

  Let $V$ be the set of variables introduced on \cref{line:V}. It remains to
  prove that repeating the transformation
  \[
    \forall X \in \Omega\text{. } \psi \mapsto \{\, \psi[x/\{\, X \,\}], \forall X \in \Omega\text{. } X \ne x \Rightarrow \psi \,\}
  \]
  for all variables $X \in V$ produces the same clauses as
  \cref{line:conditions,line:generation} (except for some tautologies that the
  latter method skips). Note that the order in which these operations are
  applied is immaterial. In other words, if we add inequality constraints for
  variables ${\{\, X_{i} \,\}}_{i=1}^{n}$, and apply substitution for variables
  $V \setminus {\{\, X_{i} \,\}}_{i=1}^{n}$, then---regardless of the order of
  operations---we get
  \[
    \forall X_{1}, \dots, X_{n} \in \Omega\text{.} \bigwedge_{i=1}^{n} X_{i} \ne x \Rightarrow \psi\left[x / \left(V \setminus {\{\, X_{i} \,\}}_{i=1}^{n}\right)\right].
  \]
  This formula is equivalent to the clause generated on \cref{line:generation}
  with $W = V \setminus {\{\, X_{i} \,\}}_{i=1}^{n}$. Thus, for every clause $c$
  of $\phi$, the new formula $\phi'$ gets a set of clauses whose conjunction is
  equivalent to $c$. Hence, the two formulas are equivalent.
\end{proof}

\begin{theorem}[Correctness of $\CR$]
  Let $\phi$ be the input formula of \cref{alg:constraintremoval}, $(\Omega, x)$
  a replaceable pair, and $\phi'$ the output formula for when $(\Omega, x)$ is
  selected on \cref{line:crconditions}. Then $\phi \equiv \phi'$, where the
  domain $\Omega'$ introduced on \cref{line:newdomain} is interpreted as
  $\Omega \setminus \{\, x \,\}$.
\end{theorem}
\begin{proof}
  Since there is a natural bijection between the clauses of $\phi$ and $\phi'$,
  we shall argue about the equivalence of each pair of clauses. Let $c$ be an
  arbitrary clause of $\phi$ and $c'$ its corresponding clause of $\phi'$.

  If $c$ has no variables with domain $\Omega$, then it cannot have any
  constraints involving $x$, so $c' = c$. Otherwise, for notational simplicity,
  let us assume that $X$ is the only variable in $c$ with domain $\Omega$ (the
  proof for an arbitrary number of variables is virtually the same). By
  \cref{def:replaceable}, we can rewrite $c$ as $\forall X \in \Omega\text{.
  } X \ne x \Rightarrow \psi$, where $\psi$ is a formula with $X$ as the only free
  variable and with no mention of either $x$ or $\Omega$. Then
  $c' \equiv \forall X \in \Omega'\text{. } \psi$. Since
  $\Omega' \coloneqq \Omega \setminus \{\, x \,\}$, we have that $c \equiv c'$.
  Since $c$ was an arbitrary clause of $\phi$, this completes the proof that
  $\phi \equiv \phi'$.
\end{proof}

\begin{theorem}[Correctness of $\Reff$]
  Let $\phi$ be the formula used as input to \cref{alg:trycache}. Let $\psi$ be
  any formula selected on \cref{line:selectformula} of the algorithm s.t.\
  $\rho \ne \nulll$ on \cref{line:rho}. Let $\sigma$ be a domain size function.
  Then the set of models of $(\psi, \sigma \circ \rho)$ is equal to the set of
  models of $(\phi, \sigma)$.
\end{theorem}
\begin{proof}
  We first show that the right-hand side of \cref{eq:subset} is the same for
  models of both $(\phi, \sigma)$ and $(\psi, \sigma \circ \rho)$.
  \Cref{alg:trycache} ensures that $\psi \equiv \phi$ up to domains. In
  particular, this means that $\Preds(\psi) = \Preds(\phi)$. The square in
  \begin{equation}\label[diagram]{diagram:ref}
    \begin{tikzcd}
      \Preds(\psi) \ar[equal]{r} \ar[d, swap, "\pi_{\psi}"] & \Preds(\phi) \ar[d, "\pi_{\phi}"] \\
      {\Doms(\psi)}^{\ast} \ar[r, "\rho"] & {\Doms(\phi)}^{\ast} \ar[r, "\sigma"] & \mathbb{N}_{0}^{\ast}
    \end{tikzcd}
  \end{equation}
  commutes by \cref{def:formula} and the definition of $\rho$. In other words,
  for any predicate $\predicate/n \in \Preds(\psi) = \Preds(\phi)$, function
  $\rho$ translates the domains associated with $\predicate$ in $\psi$ to the
  domains associated with $\predicate$ in $\phi$. Let
  $\pi_{\phi}(\predicate) = {(\Delta_{i})}_{i=1}^{n}$ and
  $\pi_{\psi}(\predicate) = {(\Gamma_{i})}_{i=1}^{n}$ for some domains
  $\Delta_{i} \in \Doms(\phi)$ and $\Gamma_{i} \in \Doms(\psi)$. Since
  $\rho(\Gamma_{i}) = \Delta_{i}$ for all $i = 1, \dots, n$ by the commutativity
  in \cref{diagram:ref}, we get that
  \begin{equation}\label{eq:equality}
    \prod_{i=1}^{n}[\sigma(\Delta_{i})] = \prod_{i=1}^{n}[\sigma \circ \rho(\Gamma_{i})]
  \end{equation}
  as required. Obviously, any subset of the left-hand side of \cref{eq:equality}
  (e.g., $\mathfrak{M}(\predicate)$, where $\mathfrak{M}$ is a model of
  $(\phi, \sigma)$) is a subset of the right-hand side and vice versa. Since
  $\phi$ and $\psi$ semantically only differ in what domains they refer to,
  every model of $\phi$ satisfies $\psi$ and vice versa, completing the proof.
\end{proof}

\section{Solutions Found by \textsc{Crane}}\label{sec:solutions}

Here we list the exact function definitions produced by \textsc{Crane} for all
of the problem instances in Section~6, both before and after algebraic
simplification (excluding multiplications by one). The correctness of all of
them has been checked by identifying suitable base cases and verifying the
numerical answers across a range of domain sizes.

\begin{enumerate}
  \item A $\Theta(m)$ solution for counting $\Gamma \to \Gamma$ functions:
        \[
        f(m) = {\left(-1 + \sum_{l=0}^{m} \binom{m}{l} [l < 2]\right)}^{m} = m^{m}.
        \]
  \item A $\Theta(m^3 + n^3)$ solution for counting $\Gamma \to \Delta$
        surjections:
        \begin{align*}
          f(m, n) ={}& \sum_{l=0}^{m} \binom{m}{l}{(-1)}^{m-l} \sum_{k=0}^{n} \binom{n}{k} {(-1)}^{n-k}\\
                     &{\left( \sum_{j=0}^{k} \binom{k}{j} [j < 2] \right)}^{l} \\
          ={}& \sum_{l=0}^{m} \binom{m}{l}{(-1)}^{m-l}\\
                     &\sum_{k=0}^{n} \binom{n}{k} {(-1)}^{n-k} {(k+1)}^{l}.
        \end{align*}
  \item A $\Theta(m^{3})$ solution for counting $\Gamma \to \Gamma$ surjections:
        \begin{align*}
          f(m) ={}& \sum_{l=0}^{m} \binom{m}{l}{(-1)}^{m-l} \sum_{k=0}^{m} \binom{m}{k} {(-1)}^{m-k}\\
                  &{\left( \sum_{j=0}^{k} \binom{k}{j} [j < 2] \right)}^{l} \\
          ={}& \sum_{l=0}^{m} \binom{m}{l}{(-1)}^{m-l}\\
                  &\sum_{k=0}^{m} \binom{m}{k} {(-1)}^{m-k} {(k+1)}^{l}.
        \end{align*}
  \item A $\Theta(mn)$ solution for counting $\Gamma \to \Delta$ injections and
        partial injections (with different base cases):
        \begin{align*}
          f(m, n) ={}& \sum_{l=0}^m \binom{m}{l} [l<2] f(m-l, n-1)\\
          ={}& f(m, n-1) + mf(m-1, n-1).
        \end{align*}
  \item A $\Theta(m^{3})$ solution for counting $\Gamma \to \Gamma$ injections:
        \begin{align*}
          f(m) ={}& \sum_{l=0}^{m} \binom{m}{l} {(-1)}^{m-l} g(m, l); \\
          g(m, l) ={}& \sum_{k=0}^{l} \binom{l}{k} [k < 2] g(m - 1, l - k)\\
          ={}& g(m - 1, l) + lg(m - 1, l - 1).
        \end{align*}
  \item A $\Theta(m)$ solution for counting $\Gamma \to \Delta$ bijections:
        \[
        f(m, n) = mf(m-1, n-1).
        \]
  \item A $\Theta(l+mn)$ solution for counting $|\Lambda|$ partial injections
        $\Gamma \to \Delta$:
        \begin{align*}
          f(l, m, n) ={}& {g(m, n)}^{l}\\
          g(m, n) ={}& \sum_{k=0}^{n}\binom{n}{k} [k < 2] g(m-1, n-k)\\
          ={}& g(m-1, n) + ng(m-1, n-1).
        \end{align*}
\end{enumerate}

\bibliographystyle{kr}
\bibliography{paper}

\end{document}